\documentclass[letterpaper,10pt,conference]{ieeeconf}
\IEEEoverridecommandlockouts
\usepackage{amsmath,amssymb,mathtools}
\usepackage{bm}
\usepackage{cite}

\newcommand{\dd}{\,\mathrm{d}}

\usepackage{amsmath,amssymb,mathtools}
\usepackage{bm}
\usepackage{graphicx}
\usepackage{cite}
\usepackage{epstopdf}
\usepackage[normalem]{ulem}
\usepackage[hidelinks]{hyperref}
\usepackage{psfrag,color}
\usepackage{algorithm}
\usepackage{algpseudocode}
\usepackage{siunitx}

\newtheorem{theorem}{Theorem}

\newtheorem{remark}{Remark}
\usepackage{lmodern} % Allows infinitely scalable fonts

\title{\LARGE \bf  Enhancing Energy Harvesting in Harmonically Forced Oscillators via Regenerative Impulsive Braking }

\author{Nilay~Kant%
\thanks{The author is with the Department of Mechanical and Aerospace Engineering, Missouri University of Science and Technology, Rolla, MO, USA.}%
\thanks{Corresponding author: Nilay Kant (email: nilaykant@mst.edu)}%
}

\begin{document}
\maketitle
\begin{abstract}
This paper investigates whether regenerative impulsive braking can enhance passive energy harvesting in harmonically forced oscillators. An impulsive braking event directly extracts kinetic energy while simultaneously inducing a transient that alters subsequent passive harvesting. The analysis shows that suitably timed braking can produce a net energy gain over passive harvesting alone. The braking-induced transient is derived analytically and used to obtain an exact expression for the harvested-energy gain over an arbitrary finite post-braking horizon, together with conditions for positive gain and the optimal braking instant. The corresponding infinite-horizon analysis shows that optimal braking occurs at zero crossings of the harmonic excitation, with complete braking maximizing the gain. Numerical results further show that the finite-horizon harvested-energy gain can exceed its infinite-horizon value. Simulations using short-duration high-gain braking closely match analysis, supporting potential of the proposed approach for energy-harvesting applications.
\end{abstract}

\section{Introduction}

Vibration energy harvesting converts environmental or externally induced motion into electrical energy across a wide range of scales \cite{stephen2006energy,cai2022nexus}. Applications include ocean waves \cite{falnes2007review,viet2016energy}, railways and vehicles \cite{xie2015mathematical,zhang2016high,zhang2016portable}, civil infrastructure \cite{jung2011energy,scruggs2010multi}, human motion \cite{starner2004human}, and low-power devices such as wireless sensors, MEMS, and piezoelectric systems \cite{paradiso2005energy,tao2016novel}. A common objective is therefore to enhance harvested-energy.

Vibration energy harvesting methods can be broadly classified as passive, active, or semi-active. Passive harvesters rely on fixed mechanical or electrical dissipation, whereas active methods can reshape the response more aggressively but require powered actuation \cite{yan2019low,scapolan2016energy}. Semi-active methods provide an intermediate alternative by modulating quantities such as damping or stiffness using comparatively low-power actuation. Examples include periodic time-varying damping \cite{di2013energy}, electromagnetic damping control \cite{caruso2018semi}, latching control in wave energy converters \cite{BABARIT200677}, and switching between discrete damping levels \cite{potter2010generalisation}. These approaches show that appropriately timed changes in the system dynamics can improve energy extraction, although optimal policies are often obtained numerically.

This paper investigates regenerative impulsive braking as a semi-active harvesting mechanism for a harmonically forced oscillator. A braking event directly recovers kinetic energy through an abrupt reduction in velocity while simultaneously inducing a transient that alters subsequent passive harvesting. The net benefit therefore depends jointly on the braking instant, braking strength, directly recovered energy, and the resulting transient. A related approach, termed impulsive parametric damping, was studied in \cite{tehrani2016impulsive}, where prescribed impulsive damping sequences were shown to enhance energy extraction. However, a closed-form characterization of the incremental harvested-energy gain relative to passive harvesting and its optimal braking conditions was not provided. This paper derives the braking-induced transient and obtains exact finite- and infinite-horizon harvested-energy gains, together with conditions for net energy gain and optimal braking. More broadly, impulsive inputs have been used for enlarging regions of attraction \cite{kant2019estimation}, stabilization and orbital regulation \cite{kant2018impulsive,kant2020orbital}, and impulsive braking for energy-based orbital stabilization \cite{kant2021stabilization}, including experimental implementations using standard actuators and sensors \cite{kant2024EHGO,kant2019estimation}. The present work extends this control modality to energy harvesting by using regenerative impulsive braking to extract energy while reshaping the subsequent forced response. For the infinite-horizon case, the analysis shows that optimal braking occurs at zero crossings of the external force, with complete braking maximizing the gain whenever positive gain is achievable. Numerical results further reveal that the finite-horizon gain can vary nonmonotonically with the evaluation horizon and exceed its infinite-horizon value. Finally, dynamic simulations demonstrate that the ideal impulsive braking action can be closely approximated using short-duration high-gain braking.

The remainder of this paper is organized as follows. Section~II formulates the problem. Section~III derives the braking-induced transient response. Section IV develops the finite- and infinite-horizon harvested-energy analysis and optimal braking conditions. Section~V provides numerical and dynamic simulation results, followed by concluding remarks in Section~VI.

\section{Problem Formulation} \label{sec:problem}

Consider the harmonically forced underdamped linear oscillator shown in Fig. \ref{fig1}, whose dynamics is given by:
\begin{equation}
m\ddot{x}+c\dot{x}+kx=F_0\sin(\omega t)
\label{eq:plant}
\end{equation}

\noindent where $m>0$ is the mass, $c>0$ is the damping coefficient, $k>0$ is the spring constant, and $F_0>0$ and $\omega>0$ are the amplitude and frequency of the external harmonic excitation respectively. The damper acts as a passive energy harvesting mechanism that converts the motion of the mass into harvested-energy.  The steady-state displacement and velocity of the mass are \cite{thomson2008theory}:
\begin{subequations}
\label{eq:steady_state}
\begin{align}
x_p(t)&=X\sin(\omega t-\phi) \label{eq:steady_x}\\
\dot{x}_p(t)&=X\omega\cos(\omega t-\phi)
\label{eq:steady_v}
\end{align}
\end{subequations}

\noindent where
\begin{subequations}
\label{eq:steady_parameters}
\begin{align}
X(\omega)&=\frac{F_0}{\sqrt{(k-m\omega^2)^2+(c\omega)^2}}
\label{eq:steady_amplitude}\\
\phi(\omega)&=\operatorname{atan2}\!\left(c\omega, k-m\omega^2\right)
\label{eq:steady_phase}
\end{align}
\end{subequations}

\begin{figure}[t!]
 \centering
\psfrag{A}[c][c][1][0]{{$m$}}
\psfrag{B}[c][c][1][0]{{$k$}}
\psfrag{C}[c][c][1][0]{ \scriptsize$F_0 \sin (\omega t)$}
\psfrag{D}[c][c][1][0]{ $c$}
\psfrag{E}[c][c][1][0]{\scriptsize Energy Harvester}
\psfrag{F}[c][c][1][0]{\scriptsize $x(t)$}
 \includegraphics[width=0.85\linewidth]{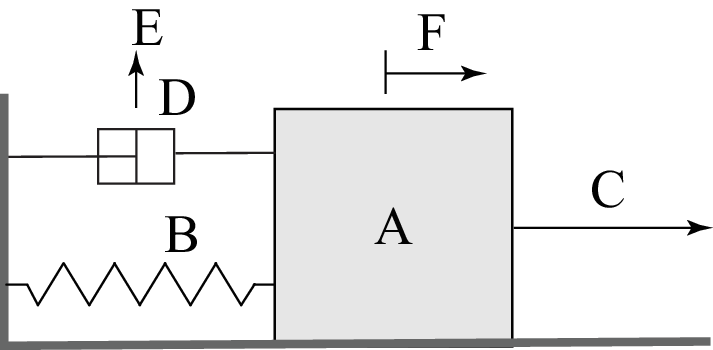}
 \caption{A harmonically forced linear oscillator.}
 \label{fig1}
 \end{figure}

We augment the continuous passive energy harvesting with an ideal regenerative impulsive brake\footnote{The ideal impulsive brake can be approximated by applying a short-duration high-gain velocity-feedback braking input, as demonstrated in Section \ref{sec5b}.} applied at a desired time $t=t_0$. Let $t_0^-$ and $t_0^+$ denote the instants immediately before and after the braking event. The impulsive braking law induces an instantaneous reduction in the velocity with no change in position, given as:
\begin{equation}
x(t_0^+)=x(t_0^-), \qquad \dot{x}(t_0^+)=\alpha\dot{x}(t_0^-), \quad 0\leq\alpha\leq1
\label{eq:reset}
\end{equation}

\noindent where $\alpha$ characterizes the fraction of pre-braking velocity retained immediately after braking. Denote the pre-braking steady-state velocity by
\begin{equation}
v_0\triangleq\dot{x}_p(t_0)=X\omega\cos(\omega t_0-\phi)
\label{eq:v0}
\end{equation}

\noindent The energy harvested through ideal regenerative impulsive braking equals the reduction in kinetic energy of the mass, given as follows
\begin{equation}
W_b=\frac{1}{2}m\left[\dot{x}^2(t_0^-)-\dot{x}^2(t_0^+)\right]=\frac{1}{2}m(1-\alpha^2)v_0^2
\label{eq:brake_energy}
\end{equation}

\noindent \textbf{Problem Statement:} For a fixed post-braking time horizon $T>0$, we compare the total energy extracted over $[t_0,t_0+T]$ under regenerative impulsive braking with that obtained from passive damping alone. The corresponding extracted energies are defined respectively as:
\begin{equation}
W_{\mathrm{imp}}(T)\triangleq W_b+\int_{t_0}^{t_0+T}c\dot{x}_{\mathrm{imp}}^2(t)\dd t
\label{eq:Wimpulsive}
\end{equation}
and
\begin{equation}
W_{\mathrm{pas}}(T)\triangleq\int_{t_0}^{t_0+T}c\dot{x}_p^2(t)\dd t
\label{eq:Wpassive}
\end{equation}

\noindent Furthermore, the net gain in harvested-energy due to regenerative impulsive braking is
\begin{equation}
\Delta W(T)\triangleq W_{\mathrm{imp}}(T)-W_{\mathrm{pas}}(T)
\label{eq:deltaW_def}
\end{equation}

\noindent The objective of this paper is to derive $\Delta W(T)$ analytically, obtain the conditions under which $\Delta W(T)>0$, and determine the optimal braking time $t_0$ that maximizes $\Delta W(T)$.
% ================================================================
\section{Braking-Induced Transient Response}

\label{sec:homogeneous}

\noindent Since the oscillator in \eqref{eq:plant} is in steady state before the braking event, its total response is given by \eqref{eq:steady_x}. The harmonic excitation remains unchanged after braking; therefore, braking introduces only an additional  transient response. To compute this, define the shifted time
\begin{equation}
\tau\triangleq t-t_0
\end{equation}
and write the post-braking displacement as
\begin{equation}
x_{\mathrm{imp}}(t_0+\tau)=x_p(t_0+\tau)+z(\tau)
\label{eq:trajectory_decomposition}
\end{equation}

\noindent where $z(\tau)$ denotes the braking-induced transient response.

After the braking, both $x_{\mathrm{imp}}$ and $x_p$ are governed by \eqref{eq:plant}. Substituting \eqref{eq:trajectory_decomposition} into \eqref{eq:plant} and subtracting the corresponding equation satisfied by $x_p(t)$ cancels the common forcing term $F_0\sin(\omega t)$, yielding the homogeneous system
\begin{equation}
m\ddot{z}+c\dot{z}+kz=0
\label{eq:z_dynamics}
\end{equation}

\noindent At $\tau=0$, in accordance with \eqref{eq:reset}, position remains unchanged but the velocity changes from $v_0$ in \eqref{eq:v0} to $\alpha v_0$. Hence, the initial conditions for \eqref{eq:z_dynamics} are given by:
\begin{equation}
z(0)=0, \qquad \dot{z}(0)=-(1-\alpha)v_0
\label{eq:z_initial}
\end{equation}

\noindent Define
\begin{equation}
a \triangleq c/(2m), \quad \omega_n \triangleq \sqrt{k/m}, \quad \omega_d\triangleq\sqrt{\omega_n^2-a^2}
\label{eq:dynamic_parameters}
\end{equation}

\noindent where $\omega_n$ and $\omega_d$ are the undamped and damped natural frequencies, and $a$ is the decay rate. Using \eqref{eq:z_initial}, the brake-induced transient which is the solution of \eqref{eq:z_dynamics} can be obtained as follows:
\begin{subequations}
\label{eq:z_response}
\begin{align}
z(\tau)&=-\frac{(1-\alpha)v_0}{\omega_d}e^{-a\tau}\sin(\omega_d\tau)
\label{eq:z_solution}\\
\dot{z}(\tau)&=-(1-\alpha)v_0h(\tau)
\label{eq:zdot_solution}\\
h(\tau)&\triangleq e^{-a\tau}\left[\cos(\omega_d\tau)-\frac{a}{\omega_d}\sin(\omega_d\tau)\right]
\label{eq:h_def}
\end{align}
\end{subequations}

\noindent Finally, using \eqref{eq:steady_v}, \eqref{eq:trajectory_decomposition} and \eqref{eq:z_response}, we get:
\begin{equation}
\dot{x}_{\mathrm{imp}}(t_0+\tau)=X\omega\cos\!\left(\omega t_0+\omega\tau-\phi\right)-(1-\alpha)v_0h(\tau)
\label{eq:post_brake_velocity}
\end{equation}

\noindent The post-braking velocity in \eqref{eq:post_brake_velocity} will now be used to quantify the change in harvested-energy produced due to impulsive braking - this is developed next.

\section{Optimal Harvested-Energy Gain}
\label{sec:energy_gain}

\subsection{Finite-Horizon Harvested-Energy Gain}\label{sec:finite_horizon}
Differentiating \eqref{eq:trajectory_decomposition} gives $\dot{x}_{\rm imp}=\dot{x}_p+\dot z$. Substituting this relation into (7) and subtracting (8),   the common passive contribution cancels, yielding
\begin{equation}
\Delta W(T)=W_b+2c\int_0^T\dot{x}_p(t_0+\tau)\dot{z}(\tau)\dd\tau+c\int_0^T\dot{z}^2(\tau)\dd\tau
\label{eq:deltaW_exact}
\end{equation}

\noindent For any finite $T>0$, the value of $\Delta W(T)$ in \eqref{eq:deltaW_exact} is exact and does not require the braking-induced transient to be negligible. We aim to characterize the dependence of $\Delta W(T)$ on the braking instant $t_0$.

Let
\begin{equation}
r\triangleq1-\alpha, \qquad \psi\triangleq\omega t_0-\phi
\label{eq:r_psi}
\end{equation}

\noindent such that from \eqref{eq:v0}, $v_0=X\omega\cos\psi$. Using \eqref{eq:steady_v} and \eqref{eq:zdot_solution}, we get
\[
\dot{x}_p(t_0+\tau)=X\omega\cos(\psi+\omega\tau), \quad \dot{z}(\tau)=-rv_0h(\tau)
\]

\noindent Furthermore, define the time horizon-dependent quantities:
\begin{align}\label{eq:ABC_def}
A_T&\triangleq\int_0^T h(\tau)\cos(\omega\tau)\dd\tau, \quad
B_T\triangleq\int_0^T h(\tau)\sin(\omega\tau)\dd\tau\\
C_T&\triangleq\int_0^T h^2(\tau)\dd\tau, \notag
\end{align}

\noindent and let $\eta\triangleq a/\omega_d$, $\nu_{\pm}\triangleq\omega\pm\omega_d$, and the auxiliary integrals
\begin{equation}
\begin{aligned}
\mathcal{C}_{\lambda}(\nu,T)
&\triangleq\int_0^T e^{-\lambda\tau}\cos(\nu\tau)\dd\tau\\
&=\frac{\lambda\left[1-e^{-\lambda T}\cos(\nu T)\right]+\nu e^{-\lambda T}\sin(\nu T)}{\lambda^2+\nu^2}\\
\mathcal{S}_{\lambda}(\nu,T)
&\triangleq\int_0^T e^{-\lambda\tau}\sin(\nu\tau)\dd\tau\\
&=\frac{\nu\left[1-e^{-\lambda T}\cos(\nu T)\right]-\lambda e^{-\lambda T}\sin(\nu T)}{\lambda^2+\nu^2}
\end{aligned}
\label{eq:aux_integrals}
\end{equation}

\noindent Using \eqref{eq:h_def} and the standard product to sum trigonometric identities, \eqref{eq:ABC_def} can be expressed as follows:
\begin{equation}
\begin{aligned}
A_T={}&\frac{1}{2}\left[\mathcal{C}_a(\nu_-,T)+\mathcal{C}_a(\nu_+,T)\right]\\
&-\frac{\eta}{2}\left[\mathcal{S}_a(\nu_+,T)-\mathcal{S}_a(\nu_-,T)\right]\\
B_T={}&\frac{1}{2}\left[\mathcal{S}_a(\nu_-,T)+\mathcal{S}_a(\nu_+,T)\right]\\
&-\frac{\eta}{2}\left[\mathcal{C}_a(\nu_-,T)-\mathcal{C}_a(\nu_+,T)\right]\\
C_T={}&\frac{1+\eta^2}{2}\mathcal{C}_{2a}(0,T)
+\frac{1-\eta^2}{2}\mathcal{C}_{2a}(2\omega_d,T)\\
&-\eta\mathcal{S}_{2a}(2\omega_d,T)
\end{aligned}
\label{eq:ABC_closed}
\end{equation}

\noindent As expected, the energy gain for time horizon $T$ depends on the decay rate $a$, the excitation frequency $\omega$, and the damped natural frequency $\omega_d$.

\begin{theorem}
\label{thm:finite_gain}
For the oscillator \eqref{eq:plant} subject to the regenerative braking law \eqref{eq:reset}, the net gain in harvested-energy over a fixed post-braking horizon $T>0$ is given by
\begin{equation}
\Delta W(T)=X^2\omega^2\left[P_T\cos^2\psi+Q_T\sin\psi\cos\psi\right]
\label{eq:finite_gain}
\end{equation}
where
\begin{equation}
\begin{aligned}
P_T&\triangleq\frac{1}{2}m(1-\alpha^2)-2crA_T+cr^2C_T\\
Q_T&\triangleq2crB_T
\end{aligned}
\label{eq:PQ}
\end{equation}
Consequently, impulsive braking provides a net gain in harvested-energy whenever
\begin{equation}
P_T\cos^2\psi+Q_T\sin\psi\cos\psi>0
\label{eq:finite_benefit_condition}
\end{equation}

\noindent Furthermore, for $R_T\triangleq\sqrt{P_T^2+Q_T^2}>0$, the energy gain is maximized when $t_0 = t_0^\star$ satisfies the following relation
\begin{equation}
\omega t_0^\star=\phi+\frac{1}{2}\operatorname{atan2}(Q_T,P_T)+n\pi, \quad n\in\mathbb{Z}
\label{eq:t0_opt}
\end{equation}

\noindent with the corresponding maximum harvested-energy gain:
\begin{equation}
\Delta W_{\max}(T)=\frac{X^2\omega^2}{2}\left[P_T+\sqrt{P_T^2+Q_T^2}\right]
\label{eq:finite_max_gain}
\end{equation}
\end{theorem}

\begin{proof}
Using $\cos(\psi+\omega\tau)=\cos\psi\cos(\omega\tau)-\sin\psi\sin(\omega\tau)$, the two integral terms in \eqref{eq:deltaW_exact} become
\begin{align*}
2c\int_0^T\dot{x}_p\dot{z}\dd\tau
&=-2crX\omega v_0\left(A_T\cos\psi-B_T\sin\psi\right)\\
c\int_0^T\dot{z}^2\dd\tau
&=cr^2v_0^2C_T
\end{align*}
Substituting these expressions and $W_b$ from \eqref{eq:brake_energy} into \eqref{eq:deltaW_exact}, and using $v_0=X\omega\cos\psi$ from \eqref{eq:v0} and \eqref{eq:r_psi} gives
\[
\begin{aligned}
\Delta W(T)&=\\
X^2\omega^2&\Bigg[
\left(\frac{1}{2}m(1-\alpha^2)-2crA_T+cr^2C_T\right)\cos^2\psi\\
&+2crB_T\sin\psi\cos\psi
\Bigg]
\end{aligned}
\]
which gives \eqref{eq:finite_gain} using \eqref{eq:PQ}. Since $X^2\omega^2>0$, the condition $\Delta W(T)>0$ gives \eqref{eq:finite_benefit_condition}.

Next, using $\cos^2\psi=(1+\cos2\psi)/2$ and $\sin\psi\cos\psi=\sin2\psi/2$, \eqref{eq:finite_gain} can be written as
\[
\Delta W(T)=\frac{X^2\omega^2}{2}\left[P_T+P_T\cos(2\psi)+Q_T\sin(2\psi)\right]
\]
Let $\gamma_T\triangleq\operatorname{atan2}(Q_T,P_T)$. Then
\[
P_T\cos(2\psi)+Q_T\sin(2\psi)=R_T\cos(2\psi-\gamma_T)
\]
and therefore
\[
\Delta W(T)=\frac{X^2\omega^2}{2}\left[P_T+R_T\cos(2\psi-\gamma_T)\right]
\]

\noindent For $R_T>0$, the maximum occurs when $2\psi-\gamma_T=2n\pi$. In the above equation, using $\psi=\omega t_0-\phi$ gives \eqref{eq:t0_opt}, and setting $\cos(2\psi-\gamma_T)=1$ gives \eqref{eq:finite_max_gain}.
\end{proof}

\begin{remark}
If $R_T=0$, then $P_T=Q_T=0$ which implies the net harvested-energy gain $\Delta W(T)=0$. In general, however, the optimal braking instant will depend on the horizon $T$, the braking ratio $\alpha$, and on $P_T$ and $Q_T$ which effectively depend on the oscillator response.
\end{remark}

\subsection{Infinite-Horizon Harvested-Energy Gain}
\noindent The finite-horizon analysis in subsection \ref{sec:finite_horizon} quantifies the energy gain of a single regenerative braking event over a prescribed interval $[t_0,t_0+T]$. It is also useful to quantify the net effect of this braking event after the induced transient has completely decayed and the oscillator has returned to its periodic steady state. This corresponds to the infinite-horizon limit of the finite-horizon result. Accordingly, let the infinite-horizon gain in harvested-energy be
\begin{equation}
\Delta W_\infty\triangleq\lim_{T\rightarrow\infty}\Delta W(T)
\label{eq:deltaW_infty_def}
\end{equation}

\begin{theorem}
\label{thm:asymptotic_gain}

For the oscillator \eqref{eq:plant} subject to the regenerative braking law \eqref{eq:reset}, the infinite-horizon gain in harvested-energy is given by
\begin{equation}
\Delta W_\infty=(1-\alpha)mX^2\omega^2\left[\cos^2(\omega t_0)-\sin^2\phi\right]
\label{eq:asymptotic_gain}
\end{equation}

\noindent For any fixed braking ratio $0\leq\alpha<1$, a net gain $\Delta W_\infty>0$ is obtained if and only if
\begin{equation}
\cos^2(\omega t_0)>\sin^2\phi
\label{eq:asymptotic_benefit_condition}
\end{equation}

\noindent For the same $\alpha$, the maximum infinite-horizon gain is
\begin{equation}
\max_{t_0}\Delta W_\infty=(1-\alpha)mX^2\omega^2\cos^2\phi,
\label{eq:max_gain}
\end{equation}
and is attained at the optimal braking times $t_0^\star$ satisfying
\begin{equation}
\omega t_0^\star=n\pi, \qquad n\in\mathbb{Z}
\label{eq:zero_crossing_time}
\end{equation}
which correspond to zero crossings of the external harmonic force. Furthermore, if $\alpha$ is also optimized over $0\!\leq\!\alpha\!\leq\!1$ and $\cos\phi\neq0$, the optimal braking ratio is
\begin{equation}
\alpha^\star=0
\label{eq:alpha_opt}
\end{equation}
and the overall maximum harvested-energy is
\begin{equation}
\max_{t_0,\alpha}\Delta W_\infty=mX^2\omega^2\cos^2\phi,
\label{eq:full_brake_gain}
\end{equation}
attained at $t_0=t_0^\star$ and $\alpha=\alpha^\star$.
\end{theorem}

\begin{proof}
From \eqref{eq:aux_integrals}, we obtain
\[
\lim_{T\rightarrow\infty}\mathcal{C}_{\lambda}(\nu,T)
=\frac{\lambda}{\lambda^2+\nu^2},
\qquad
\lim_{T\rightarrow\infty}\mathcal{S}_{\lambda}(\nu,T)
=\frac{\nu}{\lambda^2+\nu^2}
\]

\noindent Using these limits in \eqref{eq:ABC_closed} gives
\begin{align*}
A_\infty=\frac{2a\omega^2}{D}, \quad
&B_\infty=-\frac{\omega(\omega_n^2-\omega^2)}{D}, \quad
C_\infty=\frac{1}{4a}\cr
D &\triangleq (\omega_n^2-\omega^2)^2+4a^2\omega^2
\end{align*}

\noindent Taking $T\rightarrow\infty$ in \eqref{eq:PQ}, and using $c=2ma$ from \eqref{eq:dynamic_parameters} and $r=1-\alpha$ from \eqref{eq:r_psi}, we get after simplification
\begin{equation}
P_\infty=rm\left(1-\frac{8a^2\omega^2}{D}\right), \quad
Q_\infty=-\frac{4mar\omega(\omega_n^2-\omega^2)}{D}
\label{eq:pq_infty}
\end{equation}

\noindent From \eqref{eq:steady_phase}, using $c=2ma$ and $k=m\omega_n^2$, we have
\[
\cos\phi=(\omega_n^2-\omega^2)/\sqrt{D},
\quad
\sin\phi=(2a\omega)/(\sqrt{D})
\]

\noindent Thus, using $\cos(2\phi)=1-2\sin^2\phi$ and $\sin(2\phi)=2\sin\phi\cos\phi$, \eqref{eq:pq_infty} simplifies to
\[
P_\infty=rm\cos(2\phi), \quad
Q_\infty=-rm\sin(2\phi)
\]

\noindent Taking $T\rightarrow\infty$ in \eqref{eq:finite_gain} and substituting the above  gives
\[
\Delta W_\infty
=rmX^2\omega^2
\left[
\cos(2\phi)\cos^2\psi
-\sin(2\phi)\sin\psi\cos\psi
\right]
\]

\noindent Factoring $\cos\psi$ and using the angle-addition identity,
\[
\Delta W_\infty
=rmX^2\omega^2\cos\psi\cos(\psi+2\phi)
\]

\noindent Using $\psi=\omega t_0-\phi$ and $r=1-\alpha$, we have
\[
\Delta W_\infty
=(1-\alpha)mX^2\omega^2
\cos(\omega t_0-\phi)\cos(\omega t_0+\phi)
\]

\noindent and then utilizing the identity $\cos(A-B)\cos(A+B)=\cos^2A-\sin^2B$ gives \eqref{eq:asymptotic_gain}.

Next, for any fixed $0\!\leq\!\alpha\!<\!1$, the factor $(1-\alpha)mX^2\omega^2$ in \eqref{eq:asymptotic_gain} is strictly positive. Hence, $\Delta W_\infty\!>\!0$ if and only if \eqref{eq:asymptotic_benefit_condition} holds. Moreover, $\Delta W_\infty$ is maximized with respect to $t_0$ when $\cos^2(\omega t_0)=1$, which occurs when \eqref{eq:zero_crossing_time} holds. Substitution $t_0 = t_0^\star$ from \eqref{eq:zero_crossing_time} into \eqref{eq:asymptotic_gain} establishes \eqref{eq:max_gain}. Finally, if $\alpha$ is also optimized and $\cos\phi\neq0$, the maximum in \eqref{eq:max_gain} occurs at $\alpha^\star=0$, yielding \eqref{eq:full_brake_gain}.
\end{proof}

\begin{remark}
When $\omega=\omega_n$, the steady-state displacement lags the harmonic forcing by $\phi=\pi/2$ \cite{thomson2008theory}. Hence, \eqref{eq:asymptotic_gain} gives $\Delta W_\infty\leq0$ for every braking instant. Thus, regenerative impulsive braking cannot provide a positive infinite-horizon energy gain when the excitation frequency equals the undamped natural frequency.
\end{remark}

\section{Results}
\subsection{Insights from Analytical Results}\label{sec:analytical-results}
The results for the system in \eqref{eq:plant} with $m=1$ kg, $c=2$ N\,s/m, $k=100$ N/m, $F_0=1$ N, and $\omega=8$ rad/s are shown in Fig. \ref{fig2}(a)-(b). Fig. \ref{fig2}(a) shows the normalized harvested-energy gain $\Delta W(T)/(mX^2\omega^2)$ as a function of the braking phase $\omega t_0$ for $\alpha=0.5$ and several normalized post-braking horizons $T_n \triangleq T/(2\pi/\omega)$. Fig. \ref{fig2}(b) shows the corresponding maximum harvested-energy gain, obtained by optimizing the braking instant, as a function of the braking ratio $\alpha$ for the same set of horizons.

It can be observed from Fig. \ref{fig2}(a) that, for short horizons, both the harvested-energy gain and the optimal braking phase depend strongly on $T_n$, consistent with \eqref{eq:finite_gain}. As the horizon increases, the optimal braking phase approaches $\omega t_0=n\pi$, as predicted by the infinite-horizon result in \eqref{eq:zero_crossing_time}, and the finite-horizon curves approach the limiting gain in \eqref{eq:asymptotic_gain} as $T_n$ increases.  However, it is worth noting that this convergence is not monotonic. In particular, for an intermediate horizon $T_n=3$, the harvested-energy gain exceeds the infinite-horizon value and also attains a larger maximum. This behavior follows from the fact that $\Delta W(T)$ is the accumulated difference between the energy harvested under regenerative braking and that obtained through passive damping alone. Its rate of change with the post-braking horizon is given by differentiating \eqref{eq:deltaW_def} as follows:
\[
(d\Delta W(T))/(dT)=c\left[\dot{x}_{\rm imp}^2(t_0+T)-\dot{x}_p^2(t_0+T)\right]
\]

\noindent which can change sign as the braking-induced transient evolves. Thus, increasing the post-braking horizon can initially increase the harvested-energy gain and later reduce it as the transient decays before the gain converges to its infinite-horizon value, as observed in Fig. \ref{fig2}(a).

Fig. \ref{fig2}(b) shows the corresponding maximum normalized harvested-energy gain as a function of the braking ratio $\alpha$, with the braking instant $t_0$ optimized separately for each $T_n$ and $\alpha$ according to \eqref{eq:t0_opt}. As expected, the harvested-energy gain vanishes at $\alpha=1$, corresponding to no braking, while stronger braking produces a larger maximum gain. For short horizons, the dependence on $\alpha$ is nonlinear because the braking strength affects both the energy extracted directly during the braking event and the subsequent transient contribution. As $T_n$ increases, the curves approach the linear dependence on $(1\!-\!\alpha)$, predicted by \eqref{eq:max_gain}. The intermediate-horizon curves also confirm the non-monotonic behavior observed in Fig. \ref{fig2}(a), with the maximum harvested-energy gain exceeding its infinite-horizon value around $T_n=3$. The results indicate that the finite-horizon analysis captures transient harvesting effects that are not reflected by the infinite-horizon result alone and shows that the largest harvested-energy gain can occur before the braking-induced transient has completely decayed.

\begin{figure}[t!]
 \centering
\psfrag{A}[c][c][1][0]{{\fontsize{6pt}{6pt}$T_n = 0.1$}}

\psfrag{C}[c][c][1][0]{\fontsize{6pt}{6pt}$T_n = 1.0$}
\psfrag{D}[c][c][1][0]{\fontsize{6pt}{6pt}$T_n = 2.0$}
\psfrag{E}[c][c][1][0]{\fontsize{6pt}{6pt}$T_n = 3.0$}
\psfrag{G}[c][c][1][0]{\fontsize{6pt}{6pt}$T_n = 5.0$}
\psfrag{H}[c][c][1][0]{\fontsize{6pt}{6pt}$T_n \rightarrow \infty$}
\psfrag{O}[c][c][1][0]{\fontsize{6pt}{6pt}$\Delta W/(mX^2\omega^2)$ }
\psfrag{P}[c][c][1][0]{\fontsize{6pt}{6pt}$\Delta W_{\max}/(mX^2\omega^2)$}
\psfrag{I}[c][c][1][0]{\fontsize{6pt}{6pt}$\pi/2$}
\psfrag{J}[c][c][1][0]{\fontsize{6pt}{6pt}$\pi$}
\psfrag{K}[c][c][1][0]{\fontsize{6pt}{6pt}$3\pi/2$}
\psfrag{L}[c][c][1][0]{\fontsize{6pt}{6pt}$2\pi$}
\psfrag{M}[c][c][1][0]{\fontsize{6pt}{6pt}$\omega t_0$}
\psfrag{N}[c][c][1][0]{\fontsize{6pt}{6pt}$\alpha$}

 \includegraphics[width=\linewidth]{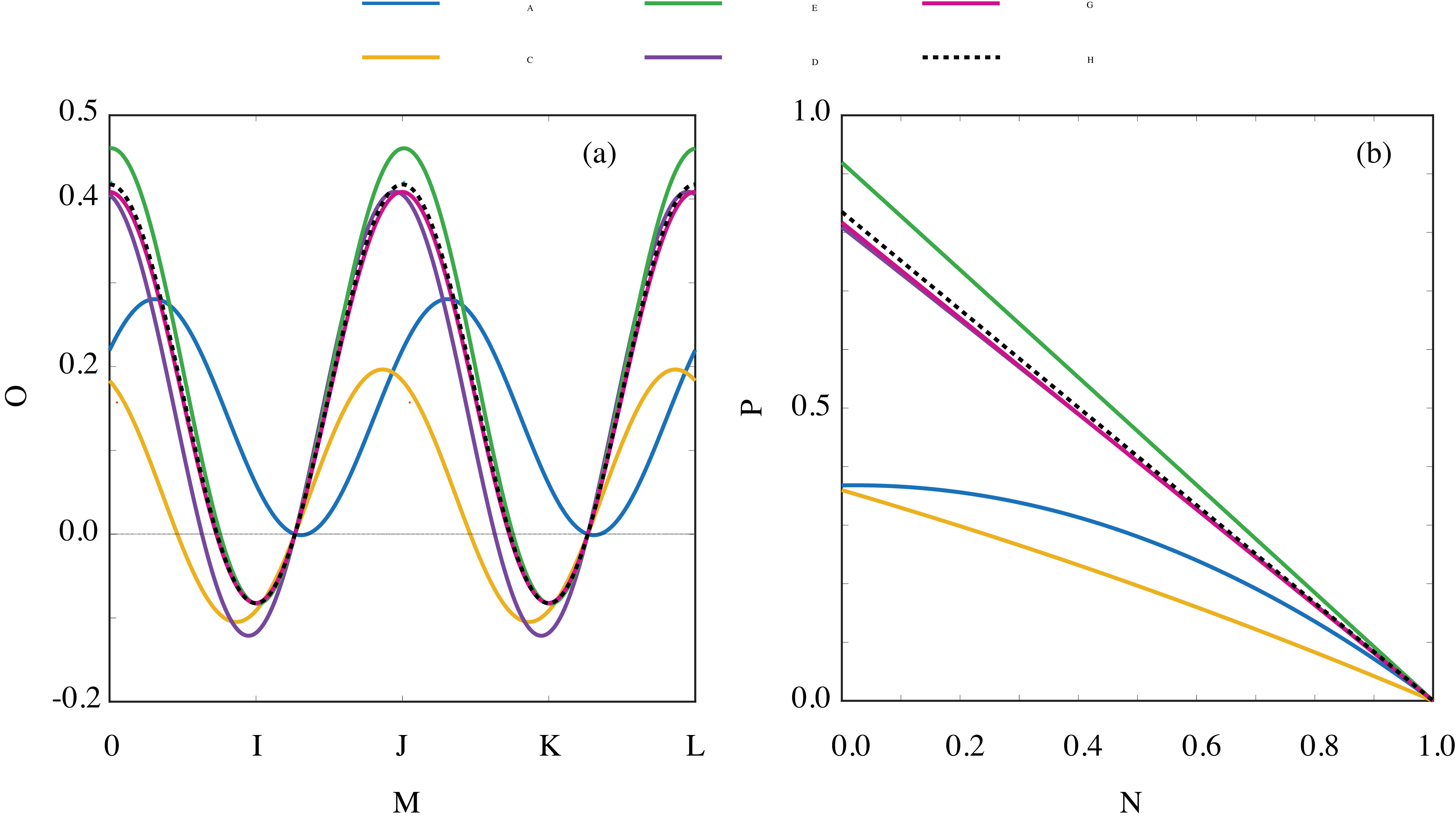}
\caption{Normalized harvested-energy gain for the harmonically forced oscillator. (a) Normalized harvested-energy gain as a function of the braking phase $\omega t_0$ for $\alpha=0.5$ and several normalized post-braking horizons $T_n$; (b) Normalized maximum harvested-energy gain, optimized over the braking instant, as a function of the braking ratio $\alpha$ for the same horizons. The dashed lines denote the corresponding infinite-horizon results.}
 \label{fig2}
 \end{figure}

\subsection{Simulation with Approximate Impulsive Braking}\label{sec5b}

To examine whether the ideal regenerative impulsive braking results can be approximated using a short-duration high-gain braking action, a dynamic simulation of \eqref{eq:plant} was performed using the same system parameters as in Section \ref{sec:analytical-results}. The regenerative brake was implemented using the high-gain feedback $u_b=-K_b\dot{x}$ with $K_b=200$ N\,s/m, following the use of high-gain feedback to approximate impulsive actions in \cite{kant2026tracking}. For comparison with the ideal braking case, the brake was activated at the optimal phase obtained from \eqref{eq:t0_opt} for $T_n=3$ and $\alpha=0$, and deactivated when $|\dot{x}|\leq0.01$ m/s. As shown in Fig. \ref{fig3}(a), the passive and braked responses are identical prior to braking since the oscillator is in periodic steady-state. Once activated, the brake reduces the velocity from $0.1869$ m/s to $0.0100$ m/s within a very short duration $\delta t\approx0.02$ s, yielding an effective braking ratio of $\alpha_{\rm eff}=0.0535$. Thus, the short-duration high-gain braking closely approximates the ideal complete braking condition ($\alpha=0$). After braking, $u_b=0$ and the system again evolves passively from the modified state, producing the transient response observed in Fig. \ref{fig3}(a), which gradually approaches the periodic steady-state.

\begin{figure}[t!]
 \centering
\psfrag{A}[c][c][1][0]{\fontsize{6pt}{7pt}\selectfont\begin{tabular}{c}regenerative\\braking\\phase\end{tabular}}
\psfrag{B}[c][c][1][0]{\fontsize{6pt}{6pt}\selectfont $\delta_t \!\approx\! 0.02$ sec}
\psfrag{C}[c][c][1][0]{\fontsize{6pt}{6pt}$\mathcal{W}_b$}
\psfrag{O}[c][c][1][0]{\fontsize{6pt}{6pt}\selectfont$\dot x$(m/s)}
\psfrag{P}[c][c][1][0]{\fontsize{6pt}{6pt}\selectfont $\mathcal{W}$(J)}
\psfrag{D}[c][c][1][0]{\fontsize{6pt}{6pt}$\mathcal{W}_{\rm imp}(t)$}
\psfrag{E}[c][c][1][0]{\fontsize{6pt}{6pt}$\mathcal{W}_{\rm pas}(t)$}
 \includegraphics[width=\linewidth]{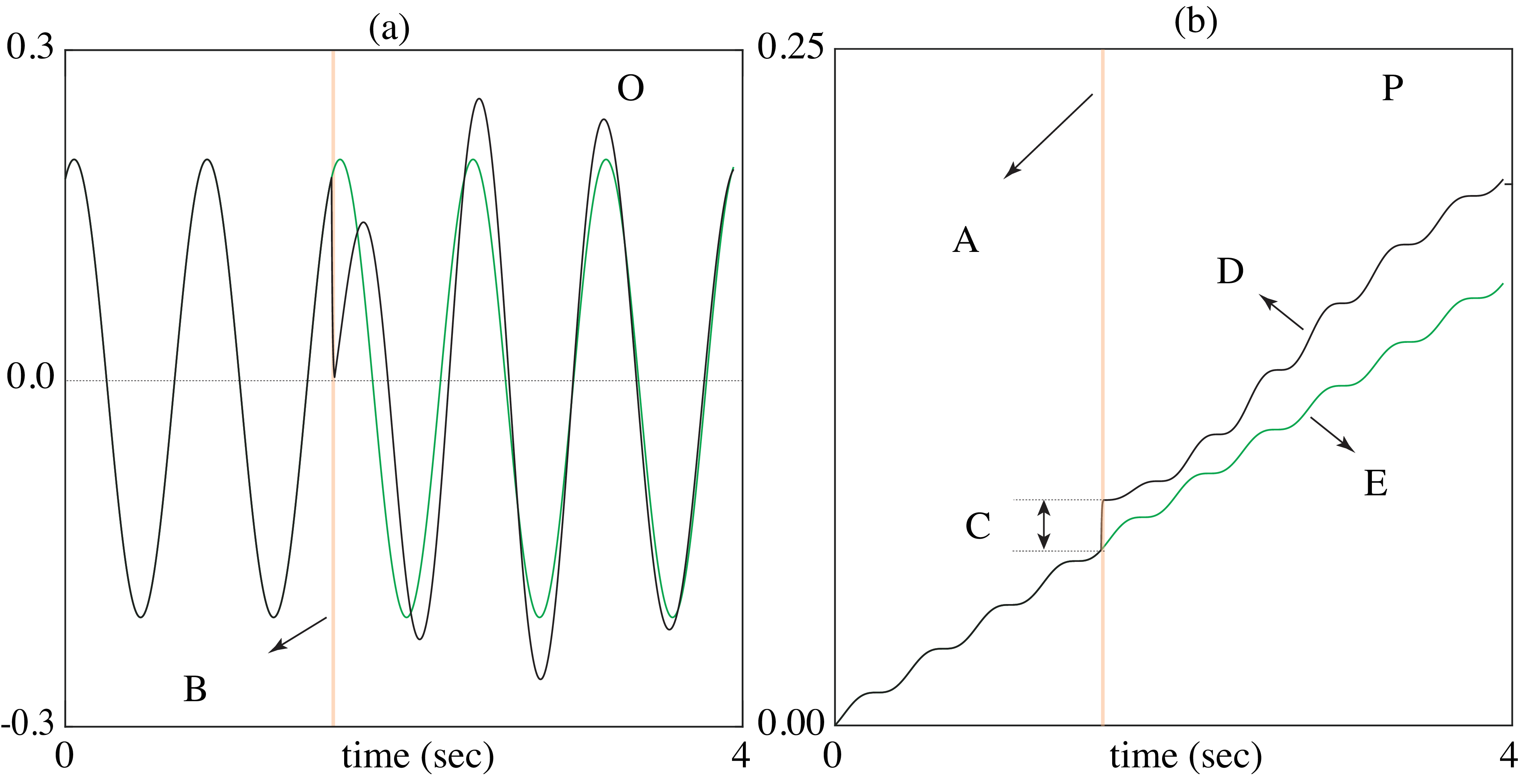}
\caption{Simulation of regenerative braking combined with passive harvesting for the system in \eqref{eq:plant}. (a) Velocity response for passive harvesting alone (green) and regenerative braking combined with passive harvesting (black). The high-gain braking phase has duration $\delta t\approx0.02$ s and rapidly reduces the velocity (shaded pink region). (b) Cumulative harvested-energies from the start of the simulation, $\mathcal{W}_{\mathrm{pas}}(t)$ (green) and $\mathcal{W}_{\mathrm{imp}}(t)$ (black). The quantity $\mathcal{W}_b$ denotes the energy extracted during high-gain braking.}
 \label{fig3}
\end{figure}

Let $\mathcal{W}_{\rm imp}(t)$ and $\mathcal{W}_{\rm pas}(t)$ denote the cumulative harvested-energies from the start of the simulation up to time $t$. This distinguishes them from $W_{\rm imp}(T)$ and $W_{\rm pas}(T)$ in \eqref{eq:Wimpulsive}-\eqref{eq:Wpassive}, which quantify the energy harvested only over the post-braking interval $[t_0,t_0+T]$. Similarly, $\mathcal{W}_b$ denotes the energy extracted by the approximate impulsive brake, whereas $W_b$ in \eqref{eq:brake_energy} denotes the instantaneous energy extraction of the ideal impulsive brake. Fig. \ref{fig3}(b) shows the cumulative harvested-energies $\mathcal{W}_{\rm imp}(t)$ and $\mathcal{W}_{\rm pas}(t)$. Prior to braking, the two curves coincide as the oscillator is in periodic steady-state. During the high-gain braking, the regenerative brake extracts an additional energy $\mathcal{W}_b$, producing the rapid increase in $\mathcal{W}_{\rm imp}(t)$. The benefit, however, extends beyond this direct extraction. Since braking modifies the system state, the induced transient alters the subsequent energy harvested through passive damping, causing $\mathcal{W}_{\rm imp}(t)$ to remain above $\mathcal{W}_{\rm pas}(t)$ after braking. Thus, the simulation reproduces the two mechanisms captured by the ideal impulsive model, which are immediate regenerative energy extraction and modification of the subsequent passive harvesting. The results indicate that the harvested-energy gain as per the theoretical results can be approximated using short-duration high-gain braking.

\section{Conclusion}

This paper investigated regenerative impulsive braking as an augmentation to passive energy harvesting in a harmonically forced linear oscillator. Starting from periodic steady-state operation, the braking-induced transient was derived analytically and used to obtain an exact expression for the harvested-energy gain over an arbitrary finite post-braking horizon. The analysis provided conditions under which braking is beneficial and determined the optimal braking instant. The corresponding infinite-horizon result showed that optimal braking occurs at zero crossings of the external harmonic force, with complete braking maximizing the harvested-energy gain whenever a positive gain is achievable. Numerical plots for the finite-horizon analysis further revealed that the harvested-energy gain need not approach its infinite-horizon value monotonically and can exceed the asymptotic harvested-energy gain at intermediate horizons. Dynamic simulations using short-duration high-gain regenerative braking closely reproduced the mechanism predicted by the ideal impulsive model, indicating that the proposed approach can be applied in practice using approximate braking actions.

The results highlight that regenerative braking can provide benefits beyond the energy extracted directly during the braking event. By modifying the system state, the braking action also reshapes the subsequent transient response and hence the energy harvested through passive damping. This suggests a broader opportunity for exploiting intermittent regenerative actions as an additional mechanism for energy extraction in mechanical systems. Future work will investigate repeated braking events and the resulting hybrid steady-state dynamics, with emphasis on maximizing average harvested power through optimal event timing and braking strength. The framework will also be extended to general periodic forcing, active harvesting and uncertain nonlinear multi-degree-of-freedom oscillators.

\bibliographystyle{IEEEtran}
\bibliography{ref}

\end{document}